\documentclass{article}
\usepackage{amsmath}
\usepackage{amsthm}
\usepackage{authblk}
\usepackage{biblatex}
\usepackage[colorlinks=true,allcolors=blue]{hyperref}
\usepackage{cleveref} % cleveref must be loaded after hyperref
\usepackage[T1]{fontenc}
\usepackage[margin=1.875in]{geometry}
\usepackage{xcolor}
\usepackage{xurl}

\DeclareMathOperator{\var}{var}
\DeclareMathOperator{\cor}{cor}
\DeclareMathOperator{\cov}{cov}
\newcommand{\Ypre}{Y_{\mathrm{pre}}}
\newcommand{\Ypost}{Y_{\mathrm{post}}}

\def\commenton{0}
\newcommand{\DTcomment}[1]{\if\commenton1{\color{red}(DT: #1)}\fi}
\newcommand{\KHcomment}[1]{\if\commenton1{\color{blue}(KH: #1)}\fi}
\newcommand{\OScomment}[1]{\if\commenton1{\color{orange}(OS: #1)}\fi}

\newtheorem{assumption}{Assumption}
\newtheorem{theorem}{Theorem}

\hypersetup{breaklinks=true}

\title{On the Limits of Regression Adjustment\footnote{Extended abstract accepted to Conference on Digital Experimentation (CODE) 2023, Cambridge, MA, USA.}}
\author{Daniel Ting}
\author{Kenneth Hung}
\affil{Meta \\ {\tt\{dting, kenhung\}@meta.com}}
\date{\today}

\begin{document}

\maketitle

\section{Introduction}
\label{sec:intro}

Regression adjustment \cite{lin2013agnostic}, sometimes known as Controlled-experiment Using Pre-Experiment Data (CUPED) \cite{deng2013improving}, is an important technique in internet experimentation. It decreases the variance of effect size estimates, often cutting confidence interval widths in half or more while never making them worse. It does so by carefully regressing the goal metric against pre-experiment features to reduce the variance.

The tremendous gains of regression adjustment begs the question: How much better can we do by engineering better features from pre-experiment data, for example by using machine learning techniques \cite{guo2021machine,guo2023generalized} or synthetic controls \cite{zhang2021regression}? Could we even reduce the variance in our effect sizes arbitrarily close to zero with the right predictors?

Unfortunately, our answer is negative. A simple form of regression adjustment, which uses just the pre-experiment values of the goal metric, captures most of the benefit. Specifically, under a mild assumption that observations closer in time are easier to predict that ones further away in time, we upper bound the potential gains of more sophisticated feature engineering, with respect to the gains of this simple form of regression adjustment. The maximum reduction in variance is $50\%$ in \Cref{thm:main}, or equivalently, the confidence interval width can be reduced by at most an additional $29\%$.

This result allows us to assess the potential value of investing in more advanced versions of regression adjustment. Furthermore, this upper bound suggests a fundamental boundary for how much analysis-based variance reduction that adjusts random imbalance from sampling can achieve. Many methods, such as inverse propensity score weighting or synthetic control, fall under this framework and suffer the same limitations. Breaking this barrier may require injecting domain knowledge on the causal mechanism which enables us to adjust for appropriate post-treatment variables \cite{deng2023variance}, or changing the experimental design directly.

\section{Main Result}
\label{sec:main-result}

When $\rho \in [0, 1]$ is the correlation between the post-experiment values and the pre-experiment values in the goal metric, the regression adjustment estimator reduces the variance by a factor of $1 - \rho^2$, i.e.\ when $\rho \approx 1$ and the post-experiment values can be nearly perfectly predicted by the pre-experiment values, then
the variance of the treatment effect estimator goes to $0$. These pre-experiment covariates can be further refined via feature engineering. For example, a machine learning model can generate predictions for the post-experiment metric, which we can use as covariates \cite{guo2021machine,guo2023generalized}.

While any covariates that are known to be unaffected by the treatment assignment is valid \cite{deng2013improving}, practitioners often restrict themselves to ``safe'' covariates that are pre-treatment \cite{xie2016improving}. Under a reasonable assumption (\Cref{asm:main}) that we cannot predict the goal metric in the post-experiment period any better than in the pre-experiment period, we show that this ``safe'' form of advanced regression adjustment is limited in further variance reduction. Specifically,
\[
    \var(\hat\delta_{\text{advanced r.a.}}) \ge \frac{1}{1 + \rho} \var(\hat\delta_{\text{basic r.a.}}).
\]
In particular, when basic regression adjustment is very effective ($\rho \approx 1$), advanced regression adjustment via feature engineering can only reduce it further by $50\%$; when basic regression adjustment is not effective ($\rho \approx 0$), advanced regression adjustment will not be effective either.

\section{Mathematical setup}
\label{sec:math}

Suppose $\Ypre$ and $\Ypost$ are the pre-experiment and post-experiment values in the goal metric of a randomly sampled user. Then the basic regression adjustment reduces the variance of the estimated average treatment effect by
\[
    \frac{\var(\hat\delta_{\text{basic r.a.}})}{\var(\hat\delta_{\text{original}})} = 1 - \rho^2 = 1 - \cor(\Ypre, \Ypost)^2.
\]

More advanced regression adjustment methods attempt to design a covariate $X$ that is as correlated with $\Ypost$ as possible. Typical advanced regression adjustment methods include:
\begin{itemize}
    \item {\em Using multiple pre-experiment covariates.} We can consider the resulting linear combination of the multiple pre-experiment covariates as $X$.
    \item {\em Producing predictions of $\Ypost$ based on external datasets.} While the predictions can be used directly in a diff-in-diff way, \cite{guo2023generalized} recommends regressing $\Ypost$ on these predictions to maintain prediction unbiasedness, a property that gives more robustness to model misspecification. In this case, these predictions play the role of $X$.
    \item {\em Producing predictions of $\Ypost$ using $\Ypost$.} Using post-experiment data of specifically goal metric is technically not allowed in regression adjustment. \cite{guo2021machine} circumvents this by cross-fitting to avoid overfitting. We are now effectively in the previous scenario --- the model is always based on a different split, i.e.\ data external to this split.
\end{itemize}

Suppose $X$ is the best possible covariate for regression adjustment, and denote the correlation matrix of $(X, \Ypre, \Ypost)$ as follows.
\begin{equation}
    \begin{pmatrix}
        1 & \sigma & \tau \\
        \sigma & 1 & \rho \\
        \tau & \rho & 1
    \end{pmatrix}.
\label{eq:cor-mat}
\end{equation}
Our assumption in \Cref{sec:main-result} can be stated more formally as \Cref{asm:main}.
\begin{assumption}
\label{asm:main}
    The correlation between $X$ and $\Ypost$ is at most the correlation between $X$ and $\Ypre$. In other words, $\tau \le \sigma$.
\end{assumption}

The variance reduction from advanced regression adjustment is $1 - \tau^2$, so the relative variance reduction compared to basic regression adjustment is given by $(1 - \tau^2) / (1 - \rho^2)$. Our main result gives a lower bound to this ratio.
\begin{theorem}
\label{thm:main}
    Suppose $X$ is the best possible covariate for regression adjustment. Under \Cref{asm:main}, we have
    \[
        \frac{1 - \tau^2}{1 - \rho^2} \ge \frac{1}{1+\rho}.
    \]
\end{theorem}

\begin{proof}
    Since $X$ is the best possible covariate for regression adjustment, the maximum correlation between $\Ypost$ and any linear combination in form of $aX + b\Ypre$ must be attained at $(a, b) = (1, 0)$. Hence the optimization problem
    \[
        \text{maximize } \cov(aX + b\Ypre, \Ypost) \text{ subject to } \var(aX + b\Ypre) = 1
    \]
    is solved at $(a, b) = (1, 0)$. With \eqref{eq:cor-mat} and adding in $\lambda$ as a Lagrange multiplier, we can rewrite the optimization problem as
    \[
        \text{maximize } (a\tau + b\rho) - \lambda(a^2 + b^2 - 2ab\sigma - 1).
    \]
    In particular, the partial derivatives with respect to $a$ and $b$ should be $0$ at $(a, b) = (1, 0)$, so
    \[
        \tau - 2\lambda = 0 \text{ and } \rho - 2\sigma\lambda = 0 \Longrightarrow \rho = \sigma\tau.
    \]
    Finally, we have $\tau^2 \le \sigma\tau = \rho$ and hence
    \[
        \frac{1 - \tau^2}{1 - \rho^2} \ge \frac{1 - \rho}{1 - \rho^2} = \frac{1}{1+\rho}. \qedhere
    \]
\end{proof}

\section{Discussion of \texorpdfstring{\Cref{asm:main}}{Assumption 1}}

While highly intuitive, \Cref{asm:main} itself is not generally testable, in the sense that even if it holds for some $X$, there may be other covariates $X'$ that we can engineer where it no longer holds.

There are some heuristics why \Cref{asm:main} plausibly holds. For example, if we believe in some loose form of stationarity, then our ability to predict $\Ypost$ after $t$ days of treatment should be similar to predicting $\Ypre$ using data $t$ days before the experiment begins, which in turns less than our ability to predict $\Ypre$ using all available pre-experiment data.

We also note some cases where the assumption may reasonably fail. Seasonality can be one reason the assumption fails. Suppose our outcome is post-treatment time period is winter and our pre-treatment period is fall. If our goal metric measures home heating costs, then it is reasonable to believe that last winter's home heating costs ($X$) could be a better predictor of next winter's heating costs ($\Ypost$) than it is of the fall's ($\Ypre$).

Another reason is if the post-treatment time period is significantly longer than the pre-treatment period. In this case, the post-treatment period has smoothed over more daily fluctuations and may be easier to predict.

\printbibliography

@inproceedings{deng2013improving,
    author = {Deng, Alex and Xu, Ya and Kohavi, Ron and Walker, Toby},
    title = {Improving the Sensitivity of Online Controlled Experiments by Utilizing Pre-Experiment Data},
    year = {2013},
    isbn = {9781450318693},
    publisher = {Association for Computing Machinery},
    address = {New York, NY, USA},
    doi = {10.1145/2433396.2433413},
    booktitle = {Proceedings of the Sixth ACM International Conference on Web Search and Data Mining},
    pages = {123–132},
    numpages = {10},
    location = {Rome, Italy},
    series = {WSDM '13}
}

@article{lin2013agnostic,
    author = {Winston Lin},
    title = {{Agnostic notes on regression adjustments to experimental data: Reexamining Freedman’s critique}},
    volume = {7},
    journal = {The Annals of Applied Statistics},
    number = {1},
    publisher = {Institute of Mathematical Statistics},
    pages = {295 -- 318},
    year = {2013},
    doi = {10.1214/12-AOAS583}
}

@article{guo2023generalized,
    author = {Kevin Guo and Guillaume Basse},
    title = {The Generalized Oaxaca-Blinder Estimator},
    journal = {Journal of the American Statistical Association},
    volume = {118},
    number = {541},
    pages = {524-536},
    year  = {2023},
    publisher = {Taylor & Francis},
    doi = {10.1080/01621459.2021.1941053}
}

@inproceedings{guo2021machine,
    author = {Guo, Yongyi and Coey, Dominic and Konutgan, Mikael and Li, Wenting and Schoener, Chris and Goldman, Matt},
    booktitle = {Advances in Neural Information Processing Systems},
    editor = {M. Ranzato and A. Beygelzimer and Y. Dauphin and P.S. Liang and J. Wortman Vaughan},
    pages = {8637--8648},
    publisher = {Curran Associates, Inc.},
    title = {Machine Learning for Variance Reduction in Online Experiments},
    url = {https://proceedings.neurips.cc/paper_files/paper/2021/file/488b084119a1c7a4950f00706ec7ea16-Paper.pdf},
    volume = {34},
    year = {2021}
}

@inproceedings{xie2016improving,
    author = {Xie, Huizhi and Aurisset, Juliette},
    title = {Improving the Sensitivity of Online Controlled Experiments: Case Studies at Netflix},
    year = {2016},
    isbn = {9781450342322},
    publisher = {Association for Computing Machinery},
    address = {New York, NY, USA},
    doi = {10.1145/2939672.2939733},
    booktitle = {Proceedings of the 22nd ACM SIGKDD International Conference on Knowledge Discovery and Data Mining},
    pages = {645–654},
    numpages = {10},
    location = {San Francisco, California, USA},
    series = {KDD '16}
}

@unpublished{zhang2021regression,
  author = {Congshan Zhang and Dominic Coey and Matt Goldman and Brian Karrer},
  title  = {Regression Adjustment with Synthetic Controls in Online Experiments},
  month  = {11},
  year   = {2021},
  url = {https://research.facebook.com/publications/regression-adjustment-with-synthetic-controls-in-online-experiments/}
}

@inproceedings{deng2023variance,
author = {Deng, Alex and Du, Michelle and Matlin, Anna and Zhang, Qing},
title = {Variance Reduction Using In-Experiment Data: Efficient and Targeted Online Measurement for Sparse and Delayed Outcomes},
year = {2023},
publisher = {Association for Computing Machinery},
address = {New York, NY, USA},
doi = {10.1145/3580305.3599928},
booktitle = {Proceedings of the 29th ACM SIGKDD Conference on Knowledge Discovery and Data Mining},
pages = {3937–3946},
numpages = {10},
location = {Long Beach, CA, USA},
series = {KDD '23}
}

\end{document}